\documentclass[onecolumn,pra,]{revtex4-2}

\usepackage{amsmath,amsfonts,amssymb,caption,hyperref,color,epsfig,graphics,graphicx,latexsym,mathrsfs,revsymb,theorem,url,verbatim}

\hypersetup{colorlinks,linkcolor={blue},citecolor={blue},urlcolor={red}}

\usepackage{hyperref}

\makeatletter

\newcommand{\Rmnum}[1]{\expandafter\@slowromancap\romannumeral #1@}
\makeatother
\newtheorem{definition}{Definition}

\newtheorem{lemma}[definition]{Lemma}

\newtheorem{theorem}[definition]{Theorem}
\newtheorem{Corollary}[definition]{Corollary}

\newtheorem{remark}[definition]{Remark}

\def\squareforqed{\hbox{\rlap{$\sqcap$}$\sqcup$}}
\def\qed{\ifmmode\squareforqed\else{\unskip\nobreak\hfil
		\penalty50\hskip1em\null\nobreak\hfil\squareforqed
		\parfillskip=0pt\finalhyphendemerits=0\endgraf}\fi}
\def\endenv{\ifmmode\;\else{\unskip\nobreak\hfil
		\penalty50\hskip1em\null\nobreak\hfil\;
		\parfillskip=0pt\finalhyphendemerits=0\endgraf}\fi}
\newenvironment{proof}{\noindent \textbf{{Proof.~} }}{\qed}
\def\Dbar{\leavevmode\lower.6ex\hbox to 0pt
	{\hskip-.23ex\accent"16\hss}D}
\makeatletter
\def\url@leostyle{%
	\@ifundefined{selectfont}{\def\UrlFont{\sf}}{\def\UrlFont{\small\ttfamily}}}
\makeatother
\def\tr{\mathrm{tr}}

\newcommand{\bra}[1]{\langle#1|}
\newcommand{\ket}[1]{|#1\rangle}

\def\Dbar{\leavevmode\lower.6ex\hbox to 0pt
	{\hskip-.23ex\accent"16\hss}D}

\begin{document}
\title{Bounds for Revised Unambiguous Discrimination Tasks of Quantum Resources}
	
	\author{Xian Shi}\email[]
	{shixian01@gmail.com}
	\affiliation{College of Information Science and Technology,
		Beijing University of Chemical Technology, Beijing 100029, China}

	%
	
	
	
	\date{\today}
	
	
	\begin{abstract}
\indent Quantum state discrimination is a fundamental task that is meaningful in quantum information theory. In this manuscript, we consider a revised unambiguous discrimination of quantum resources. First, we present an upper bound of the success probability for a revised unambiguous discrimination task in the unasymptotic and asymptotic scenarios. Next, we generalize the task from quantum states to quantum channels. We present an upper bound of the success probability for the task under the adaptive strategy. Furthermore, we show the bound can be computed efficiently. Finally, compared with the classical unambiguous discrimination, we show the advantage of the quantum by considering a quantifier on a set of semidefinite positive operators.
 	\end{abstract}
 \maketitle
 \tableofcontents
\section{Introduction}
\label{sec:introduction}
\indent Nowadays, quantum information science attracts much attention from the whole world. One of the reasons is that people have seen that, compared with classical information science, quantum information science can process information tasks more quickly and effectively \cite{nielsen,wilde,gyong,yang2023}. Nevertheless, a deeper understanding of the features of quantum mechanics is still essential to explore better methods to cope with information tasks.

Quantum state discrimination is a topic related to several ultimate quantum properties closely \cite{bergou2010discrimination,bae2015quantum}, which has received persistent attention from the researchers of quantum information theory \cite{yuen1975optimum,chefles,bou,bergou2010discrimination,carmeli2022quantum}. A fundamental fact of quantum state discrimination is that non-orthogonal quantum states cannot be distinguished perfectly \cite{nielsen}, which associates with the fundamental \emph{no-cloning} theorem \cite{wootters1982single}. Besides, quantum state discrimination also plays an important role in many areas of quantum information science, such as quantum key distribution \cite{leverrier2009unconditional,namkung2024analysis}, quantum algorithms \cite{wu2008exact,childs2016optimal}, and quantum sensing \cite{zhuang2017optimum,pirandola2018advances}.

The topic of quantum state discrimination begins with Helstrom \cite{helstrom} and Holevo \cite{Holevo} in 1960s and 1970s. Since then, there has been further developed on the topic under different optimal strategies, such as minimum error discrimination \cite{ha2021quantum}, unambiguous ­ discrimination \cite{ivanovic1987differentiate,bergou2012optimal,herzog2012optimal},
and maximal ­confidence  \cite{croke2006maximum}. Besides, there exists other strategies for quantum state discrimination, such as sequential measurement strategies \cite{bergou2013extracting,fields2020sequential}. Recently, much attention has been paid to the topic of quantum channels under the product strategy and adaptive strategy in multi-shot scenarios \cite{harrow2010adaptive,cooney2016strong,yuan2017fidelity,pirandola2019fundamental,zhuang2020ultimate,wilde2020amortized,fang2021towards,salek2021asymptotic,salek2022usefulness,huang2024exact}. However, most of the results on the discrimination of quantum resources in the asymptotic limit are not easily computed. 

Different from quantum state discrimination, quantum exclusion can be seen as ruling out the corresponding state \cite{bandyopadhyay2014conclusive}. Recently, the task of quantum exclusion was studied under the information-theoretic treatment \cite{hsieh2023quantum,mishra2024optimal}. Moreover, Ji $et$ $al.$ in \cite{ji2024barycentric} made a connection between the upper bounds of error exponent of state exclusion and channel exclusion under the adaptive strategies in the unasymptotic and asymptotic scenarios, they also presented a computable bound for the channel exclusionunder the adaptive strategy.

In this paper, we consider a revised game on unambiguous discrimination of quantum states and quantum channels, and we present the following results on the game.
\begin{itemize}
	\item For quantum state discrimination, we present an analytical expression for the success of revised unambiguous discrimination games, we also present an upper bound of the success of revised unambiguous discrimination games in the asymptotic scenario. 
	\item For quantum channel discrimination, we mainly consider the revised unambigous discrimination games under the adaptive strategies. Specifically, we present an upper bound of the quantum channel discrimination games in the asymptotic scenario, we also show that the upper bound can be computed efficiently through a semidefinite programming.
	\item We define a function on a set of positive semidefinite operators, then by considering the relation between the quantum revised unambiguous state discrimination and the classical, we presents an operational interpretation of the function, which also shows the advantage under the quantum strategy.
\end{itemize}

\section{Preliminary Knowledge}

In the section, firstly we present a whole mathematical notations used in the manuscript, then we present the definitions and properties of divergence measures for quantum states and channels.

\subsection{Notations}
\indent Let $\mathcal{H}_A$ be the Hilbert space which is relevant to the quantum system $A$. Let $\boldsymbol{Herm}_A$ and $\boldsymbol{PSD}_A$ be the set of Hermitian operators and positive semidefinite operators acting on $\mathcal{H}_A$. Here $\gamma\ge \varphi$ denotes that $\gamma-\varphi$ is positive semidefinite. Let $supp(\vartheta)=\{\vartheta\ket{\psi}|\psi\in\mathcal{H}_A\}$ denote the support of $\vartheta$. 

A quantum state is a positive semidefinite with trace 1. Let $\boldsymbol{D}_A$ be the set of quantum states acting on $\mathcal{H}_A$, and let $\tilde{\boldsymbol{D}}_A$ be the set of Hermitian operators acting on $\mathcal{H}_A$ with trace 1, which can be seen as the affine hull of $\boldsymbol{D}_A.$

A quantum channel $\Delta_{A\rightarrow B}$ is a completely positive and trace-preserving linear map from $\mathcal{H}_A$ to $\mathcal{H}_B$, let $\boldsymbol{CP}_{A\rightarrow B}$ be the set of quantum channels from $\mathcal{H}_A$ to $\mathcal{H}_B$. Assume $\dim \mathcal{H}_A=\dim\mathcal{H}_B=d$, let $J_{\mathcal{N}}=\frac{1}{d}\sum_{i,j}\ket{i}\bra{j}\otimes \mathcal{N}(\ket{i}\bra{j})$ be the Choi oerator of a quantum channel $\mathcal{N}$.

A positive operator valued measurement (POVM) $\{M_i|i=1,2,\cdots,k\}$ is a set of positive semidefinite operators with $k$ outcomes and $\sum_i M_i=\mathbb{I}.$ When a POVM applies to a state $\rho$, the probability to get the $r$-th outcome is given by 
$P(r|\rho)=\tr M_r\rho.$ Moreover, Each POVM $\{M_i\}_{i=1}^k$ can be regarded as a channel 
\begin{align}
\Lambda(\rho)=\sum_i\tr (M_i \rho)\ket{i}\bra{i},
\end{align}
which transforms a quantum state into a state acting on a classical system. We sumarize the notations used in this manuscript in TABLE \ref{T0}.
\begin{table}
	\centering
	\begin{tabular}{c|c}
		Symbol & Definition \\\hline
		$\mathbb{Z}^{+}$ & $\{1,2,3,\cdots\}$\\
		$\boldsymbol{Herm}_A$ & $\{\delta|\delta=\delta^{\dagger},\sup\limits_{\varphi\in \mathcal{H}_A}\frac{||\delta\ket{\varphi}||}{||\varphi||}< \infty,\delta:\mathcal{H}_A\rightarrow\mathcal{H}_A,\}$\\
		$\boldsymbol{PSD}_A$ & $\{\delta\in \boldsymbol{Herm}_A|\delta\ge 0\}$\\
		$\boldsymbol{D}_A$ & $\{\delta\in \boldsymbol{PSD}_A|\tr\delta=1\}$\\
		$\tilde{\boldsymbol{D}}_A$ & $\{\delta\in \boldsymbol{Herm}_A|\tr\delta=1\}$\\
		$\boldsymbol{CP}_{A\rightarrow B}$ & $\{\mathcal{M}:\boldsymbol{Herm}_A\rightarrow \boldsymbol{Herm}_B|\mathcal{M} \textit{ is completely positive}\}$\\
		$\boldsymbol{C}_{A\rightarrow B}$ & $\{\mathcal{M}\in \boldsymbol{CP}_{A\rightarrow B}|\tr \circ\mathcal{M}=\tr \}$
	\end{tabular}
	\caption{Overview of Notations}
	\label{T0}
\end{table}

\subsection{Divergence measures for states}

Divergence measures aim to quantify the relevance of two states or more states. In the following, we will recall the definitions and properties of several generalized divergences. 

First, we recall the quantum hypothesis testing for two states, $\rho$ and $\sigma$. Assume the system is either prepared in the state $\rho$ or $\sigma$, the goal of the player is to guess which the prepared state is with the use of a two-outcome POVM $\{(E_0,E_1)|E_0+E_1=\mathbb{I},E_0,E_1\ge 0\}$, here $E_0$ and $E_1$ correspond to the states $\rho$ and $\sigma$, respectively. There are two types of mistakes, 
\begin{itemize}
	\item \emph{Type I Error:} the outcome is 1 while the state is $\rho$.
	\item \emph{Type II Error:} the outcome is 0, while the state is $\sigma$.
\end{itemize}

The goal of quantum hypothesis testing is to minimize the probability of \emph{Type II Error} under the condition that the probability of \emph{Type I Error} is bounded. More precisely, let $E_0=Q$ and $E_1=\mathbb{I}-Q$, and let $\epsilon\in(0,1)$, the optimised \emph{Type II Error} probability is defined as
\begin{align*}
\beta_{\epsilon}(\rho||\sigma)=&\min_Q \tr(Q\sigma)\\
\textit{s. t.}\hspace{3mm}&\tr(\mathbb{I}-Q)\rho\le \epsilon,\\
&0\le Q\le \mathbb{I}.
\end{align*}
Next let
\begin{align*}
\gamma_{\epsilon}(\rho||\sigma)\equiv&1-\beta_{\epsilon}(\rho||\sigma),
\end{align*}
when $\tr\sigma=1$, $\gamma_{\epsilon}(\rho||\sigma)$ can be written as 
\begin{align}
\gamma_{\epsilon}(\rho||\sigma)=&\sup_Q\tr Q\sigma\\
\textit{s. t.}\hspace{3mm}&\tr Q\rho\le \epsilon,\\
&0\le Q\le \mathbb{I}.
\end{align}
Next we introduce the following generalized hypothesis testing divergence $\tilde{D}_H^{\epsilon}(\rho||\sigma)$,
\begin{align*}
\tilde{D}_H^{\epsilon}(\rho||\sigma)=&-\ln(1-\gamma_{\epsilon}(\rho||\sigma))\\
=&\sup_Q\{-\ln(1-\tr Q\sigma)|\tr Q\rho\le \epsilon,0\le Q\le \mathbb{I}\}.
\end{align*}

Next we recall the definition of the extended sandwiched Renyi divergence, furthermore, we also present the relation between the sandwiched Renyi divergence and the hypothesis testing.

Assume $\beta$ is a Hermitian operator of the system $\mathcal{H}_A$, and $\sigma$ is a positive semidefinite operator. Let $\alpha\in[0,\infty),$ the extended sandwiched Renyi divergence is defined as \cite{muller2013quantum,wang2020alpha}
\begin{equation}
\tilde{D}_{\alpha}(\beta||\sigma)\equiv
\begin{cases}
\frac{1}{\alpha-1}\log||\sigma^{\frac{1-\alpha}{2\alpha}}\beta\sigma^{\frac{1-\alpha}{2\alpha}}||_{\alpha}^{\alpha} \hspace{5mm} \textit{if $supp(\beta) \subseteq supp(\sigma)$,}\\
+\infty \hspace{5mm} \textit{otherwise},
\end{cases}
\end{equation}
where $||\gamma||_{\alpha}=[\tr(\gamma^{\dagger}\gamma)^{\frac{\alpha}{2}}]^{\frac{1}{\alpha}}$.

\begin{lemma}\label{l3}
	Let $\rho$ be a Hermitian operator of the system $\mathcal{H}_A$ with $\tr\rho=1$, and let $\sigma$ be a state of system $\mathcal{H}_A$. Then 
	\begin{align}
	\tilde{D}_H^{\epsilon}(\rho||\sigma)\le	\tilde{D}_{\alpha}(\rho||\sigma)-\frac{\alpha}{\alpha-1}\ln(1-\epsilon).
	\end{align}
	Here $\epsilon\in(0,1)$.
\end{lemma}
Here the proof is placed in the section \ref{app1}

Another generalized divergence is the geometric Renyi divergence \cite{fang2021geometric}, which is defined as

\begin{equation}
\widehat{D}_{\alpha}(\rho||\sigma)\equiv
\begin{cases}
\frac{1}{\alpha-1}\ln\tr[\sigma(\sigma^{-\frac{1}{2}}\rho\sigma^{-\frac{1}{2}})^{\alpha}]\hspace{5mm} \textit{if $supp(\rho) \subseteq supp(\sigma)$,}\\
+\infty \hspace{5mm} \textit{otherwise},
\end{cases}
\end{equation}
here $\alpha\in (1,2]$. Then we recall the properties of $\widehat{D}_{\alpha}(\rho||\sigma)$. When $\alpha\rightarrow 1$, the geometric Renyi divergence turns into the Belavkin-Staszewski relative entropy
\begin{align}
\lim\limits_{\alpha\rightarrow1}\widehat{D}_{\alpha}(\rho||\sigma)=\overline{D}(\rho||\sigma)=\tr\rho\log(\rho^{\frac{1}{2}}\sigma^{-1}\rho^{\frac{1}{2}}).\label{bs}
\end{align}

At last, the above definitions of divergence measures for states can be generalized for quantum channels. Assume $\mathcal{N}\in \boldsymbol{C}_{A\rightarrow B}$ is a channel, and $\mathcal{M}\in \boldsymbol{CP}_{A\rightarrow B}$ is completely positive. When $\alpha\in(1,2]$, the geometric Renyi channel divergence \cite{fang2021geometric} is defined as
\begin{align}
\widehat{D}_{\alpha}(\mathcal{N}||\mathcal{M})=&\sup_{\rho\in \boldsymbol{D}_{RA}}\widehat{D}_{\alpha}(\mathcal{N}(\rho_{RA})||\mathcal{M}(\rho_{RA}))\\
=&\begin{cases}
\frac{1}{\alpha-1}\ln||\tr_B[\mathcal{J}_M^{\frac{1}{2}}(\mathcal{J}^{-\frac{1}{2}}_{\mathcal{M}}\mathcal{J}_N)^{\alpha}\mathcal{J}^{\frac{1}{2}}_{\mathcal{M}}]||_{\infty}\hspace{5mm} \textit{if $supp(\mathcal{J}_{\mathcal{N}}) \subseteq supp(\mathcal{J}_{\mathcal{M}})$,}\\
+\infty \hspace{5mm} \textit{otherwise},
\end{cases}
\end{align}

The properties of the divergence measures for states and channels are summarized in the section \ref{app1}.

Based on some divergence $\mathbf{D}(\cdot||\cdot)$, we recall the $\mathbf{D}$-radius to quantify the dissimilarity of multiple states \cite{mosonyi2021,mishra2024optimal}. Here we mainly consider the $\mathbf{D}$-radius $R^{\mathbf{D}}(\cdot)$ for a set of semidefinite positive operators $\{\rho_x\}_x$, which is defined in \cite{mishra2024optimal},  
\begin{align}
R^{\mathbf{D}}(\{\rho_x\}_x)=\sup_{\{p_x\}}\inf_{\tau\in \boldsymbol{D}(\mathcal{H})}\sum_{x}p_x\mathbf{D}(\tau||\rho_x),
\end{align}
where the infimum takes over all the probability distribution $p_i>0,\sum_i p_i=1$, and the supremum takes over all the states acting on $\mathcal{H}$.  When $\mathbf{D}(\cdot||\cdot)$ is convex and lower-semicontinuous, it can be written as
\begin{align}
R^{\mathbf{D}}(\{\rho_x\}_x)=\inf_{\tau\in \boldsymbol{D}(\mathcal{H})}\max_x\mathbf{D}(\tau||\rho_x),\label{rd}
\end{align}
where the infimum takes over all the states acting on $\mathcal{H}$, and the maximum takes over all the elements in $\{\rho_x\}_x.$
\section{Main Results}
\indent In this section, we will first consider the task of quantum state discrimination. Specifically, we present an analytical formula for the one-shot success probability, we also present an upper bound on the asymptotic success probability. Furthermore, the upper bound can be computed efficiently. Next we generalize the task of quantum state discrimination to the task of quantum channel discrimination. The method to derive the bounds is generalized from quantum states discrimination. We also present a computable bound through a semidefinite programming. At last, we define a quantity on a set of postive semidefinite operators, and present an operational interpretation of the quantity.

\subsection{Quantum state discrimination}
\indent Quantum state discrimination is significant for various quantum information processing tasks.  It can be regarded as a game played by two parties, Alice and Bob. During the game, Alice prepares a state $\rho_i$ in the set $\{\rho_x\}_{x=1}^{k}$ with its prior probability $p_i>0$, respectively, here $p_i>0$ and $\sum_ip_i=1$, $i=1,2\cdots,k$. The aim of Bob is to guess the state Alice sent by applying appropriate POVM. 

In this manuscript, we consider the following revised quantum state unambiguous discrimination. Here Bob is allowed to apply a POVM with $k+1$ elements, $\{M_i|i=0,1,2,\cdots,k\}$. When the outcome is $0,$ Bob will waive the right to guess the state Alice sent. Compared with the formal unambiguously discrimination task, the task here allows the player to make errors with a certain probability. In this task, the aim is to maximize the following success probability,
\begin{align}
P_{succ}^q(\mathcal{E},\boldsymbol{p},\eta)=&\max\sum_{x=1}^k p_x\mathrm{Tr}(\rho_xQ_x),\label{ud}\\
\textit{s. t.}\hspace{3mm}&  \sum_{x=1}^k Q_x\le (1-\eta)\mathbb{I},\nonumber\\
&Q_x\ge 0,\hspace{3mm} x=1,2,\cdots,k,0,\nonumber
\end{align}
here $\mathcal{E}=\{\rho_x\}$ is a set of quantum states, $p_x$ is the prior probability that Alice sends $\rho_x$, respectively. Evidently, when $\eta=0,$ this task reduces to a general quantum state discrimination.
In the following, we will denote the above game as $(\mathcal{E},\boldsymbol{p},\eta).$ Next we present the other equivalent expression of $P_{succ}^q(\mathcal{E},\boldsymbol{p},\eta)$.

\begin{theorem}\label{t1}
	Assume $\mathcal{E}=\{\rho_x\}_{x=1}^k$ is a set with $k$ quantum states, Alice prepares a state $\rho_x$ in $\mathcal{E}$ with its prior probability $p_x$, the optimal success probabillity for  $(\mathcal{E}=\{\rho_x\},\boldsymbol{p}=\{p_x\},\eta)$ is
	\begin{align*}
	P_{succ}^q(\mathcal{E},\boldsymbol{p},\eta).=&\min_{\tau\in \tilde{\boldsymbol{D}}_A}\gamma_{\epsilon}(\pi_X\otimes\tau_A||\rho_{XA})\nonumber\\
	=&\min_{\tau\in\tilde{\boldsymbol{D}}_A}[1-\exp(-\tilde{D}_{H}^{\epsilon}(\pi_X\otimes\tau_A||\rho_{XA}))].
	\end{align*}
	Here $\epsilon=\frac{1-\eta}{k}$ and $\pi_X=\frac{1}{k}\sum_{x=1}^k\ket{x}\bra{x}$.
\end{theorem}
\begin{proof}
	Here we denote 
	\begin{align}
	\Lambda_{XA}=\sum_{x}\ket{x}\bra{x}\otimes Q_x,
	\end{align}
	here $\{Q_x|x=1,2,\cdots,k,\emptyset\}$ is a POVM, then
	\begin{align*}
	\Lambda_{XA}\le& \mathbb{I},\\
	\tr\Lambda_{XA}\rho_{XA}=& \sum_{x=1}^k p_x\tr(Q_x\rho_x),\\
	\tr(\Lambda_{XA}(\pi_X\otimes\tau_A))=& \frac{1}{k}\sum_{x=1}^k\tr(Q_x\tau)\\
	\le& \frac{1-\eta}{k},\forall \tau\in \tilde{\boldsymbol{D}}_A,
	\end{align*}
	
	\begin{align}
	P_{succ}^q(\mathcal{E},\boldsymbol{p},\eta)=&\max\limits_{\sum_x Q_x\le (1-\eta)\mathbb{I},Q_x\ge0}\sum_x p_x\tr(\rho_xQ_x)\nonumber\\
	=&\max\limits_{\sum_x Q_x\le (1-\eta)\mathbb{I},Q_x\ge0} \tr\Lambda_{XA}\rho_{XA}\nonumber\\
	\le&\max\limits_{\Lambda\in \boldsymbol{PSD}_{XA}}\min_{\tau_A\in \tilde{\boldsymbol{D}}_A}[\tr(\Lambda_{XA}\rho_{XA})|\Lambda_{XA}\le \mathbb{I}_{XA},\tr[\Lambda_{XA}(\pi_X\otimes\tau_A)]\le \frac{1-\eta}{k}]\nonumber\\
	\le &\min_{\tau_A\in\tilde{\boldsymbol{D}}_A}\gamma_{\frac{1-\eta}{k}(\pi_X\otimes\tau_A||\rho_{XA})},\label{l1}
	\end{align}
	
	Then we prove the other inequality. 
	\begin{align}
	P_{succ}^q(\mathcal{E},\boldsymbol{p},\eta)=&\min_{\tau} \{(1-\eta)\tr \tau|\tau\ge p_x\rho_x,x=1,2,\cdots,k\}\nonumber\\
	=&\min_{\tau}\{(1-\eta)\tr \tau|\mathbb{I}\otimes \tau\ge \sum_xp_x\ket{x}\bra{x}\otimes\rho_x,\tau\in \boldsymbol{Herm}_A\}\nonumber\\
	=&\min_{\tau,\mu}\{(1-\eta)\mu|\mu \mathbb{I}\otimes \tau\ge \sum_xp_x\ket{x}\bra{x}\otimes\rho_x,\tau\ge 0,\tau\in \tilde{\boldsymbol{D}}_A\}\nonumber\\
	=& \min_{\tau,\mu}\max_{\Delta}\{(1-\eta)\mu+\tr\Delta(\sum_xp_x\ket{x}\bra{x}\otimes\rho_x-\mu\mathbb{I}\otimes \tau),\tau\in \tilde{\boldsymbol{D}}_A\}\nonumber\\
	\ge&\min_{\tau\in\tilde{\boldsymbol{D}}_A}\max_{\Delta}\min_{\mu}\{\tr(\sum_x\ket{x}\bra{x}\otimes\rho_x)\Delta+\mu(1-\eta-\tr(\mathbb{I}\otimes \tau)\Delta)\}\nonumber\\
	=&\min_{\tau\in \tilde{\boldsymbol{D}}_A}\max_{\Delta}\{\tr(\sum_x\ket{x}\bra{x}\otimes\rho_x)\Delta|\tr(\mathbb{I}\otimes \tau)\Delta\le 1-\eta\}\nonumber\\
	\ge&\min_{\tau\in \tilde{\boldsymbol{D}}_A}\max_{\Delta}\{\tr(\sum_x\ket{x}\bra{x}\otimes\rho_x)\Delta|\tr(\pi_X\otimes \tau)\Delta\le\frac{1-\eta}{k},\Delta_{XA}\le \mathbb{I}_{XA}\}\nonumber\\
	=&\min_{\tau\in \tilde{\boldsymbol{D}}_A}\gamma_{\epsilon}(\pi_X\otimes\tau_A||\rho_{XA}).\label{r1}
	\end{align}
	In the last equality, we denote $\epsilon=\frac{1-\eta}{k}$. Combing (\ref{l1}) and (\ref{r1}), we finish the proof of the first equality. Due to the monotonicity of the function $-\ln(1-x)$ for $x$, we finish the proof of the second equality.
\end{proof}

\begin{theorem}\label{t2}
	Assume $\mathcal{E}=\{\rho_x\}_{x=1}^k$ is a set with $k$ quantum states, Alice prepares a state $\rho_x$ in $\mathcal{E}$ with its prior probability $p_x$, the optimal success probabillity for  $(\mathcal{E}=\{\rho_x\},\boldsymbol{p}=\{p_x\},\eta)$ satisfies the following property, 
	\begin{align}
	-\ln(1-p_{succ}(\mathcal{E},\boldsymbol{p},\eta))\le  R^{\tilde{D}_{\alpha}}(\{\rho_x\}_x)-\frac{\alpha}{\alpha-1}(\ln(p_{min})+\ln(k-1+\eta)),
	\end{align}
	here $p_{min}$ is the minimum of all the elements in $\{p_i\}_{i=1}^k$ and $\alpha\in (1,+\infty).$
\end{theorem}

\begin{proof}
	In the proof, we denote $K_{\alpha}(\rho||\sigma)=\exp[(\alpha-1)\tilde{D}_{\alpha}(\rho||\sigma)].$	Due to the Theorem \ref{t1}, we have
	\begin{align}
	&-\ln(1-p_{succ}(\mathcal{E},\boldsymbol{p},\eta))\\=&\min_{\tau\in \tilde{\boldsymbol{D}}_A}D_H^{\epsilon}(\pi_X\otimes \tau_A||\rho_{XA})\\
	\le& \min_{\tau\in\tilde{\boldsymbol{D}}_A}[\tilde{D}_{\alpha}(\pi_X\otimes\tau_A||\rho_{XA})-\frac{\alpha}{\alpha-1}\ln(1-\epsilon)]\\
	=&\min_{\tau\in \tilde{\boldsymbol{D}}_A}[\frac{1}{\alpha-1}\ln \sum_x p_x\frac{p_x^{-\alpha}}{k^{\alpha}}K_{\alpha}(\tau||\rho_x)-\frac{\alpha}{\alpha-1}\ln(1-\epsilon)]\\
	\le&\min_{\tau\in\tilde{\boldsymbol{D}}_A}\max_x[\frac{1}{\alpha-1}\ln\frac{1}{p_x^{\alpha}k^{\alpha}}K_{\alpha}(\tau||\rho_x)-\frac{\alpha}{\alpha-1}\ln(1-\epsilon)]\\
	=&\min_{\tau\in\tilde{\boldsymbol{D}}_A}\max_x[\tilde{D}_{\alpha}(\tau||\rho_x)-\frac{1}{\alpha-1}\ln(p_x^{\alpha}k^{\alpha})-\frac{\alpha}{\alpha-1}\ln(1-\epsilon)]\\
	\le&\min_{\tau\in\tilde{\boldsymbol{D}}_A}\max_x[\tilde{D}_{\alpha}(\tau||\rho_x)-\frac{1}{\alpha-1}\ln(p_{min}^{\alpha}k^{\alpha})-\frac{\alpha}{\alpha-1}\ln(1-\epsilon)]\\
	=&\min_{\tau\in\tilde{\boldsymbol{D}}_A}\max_x[\tilde{D}_{\alpha}(\tau||\rho_x)-\frac{\alpha}{\alpha-1}(\ln(p_{min})+\ln(k-1+\eta))]\\
	=&R^{\tilde{D}_{\alpha}}(\{\rho_x\}_x)-\frac{\alpha}{\alpha-1}(\ln(p_{min})+\ln(k-1+\eta)).
	\end{align} 
	Here $\epsilon=\frac{1-\eta}{k}$, and the first inequality is due to Lemma \ref{l3}. In the third inequality, $p_{min}$ is the minimum of all the elements in $\{p_x\}_{x=1}^k$, the last equality is due to the equality (\ref{rd}).
\end{proof}

\begin{remark}
	Here we show that the upper bound of $p_{succ}(\mathcal{E},\boldsymbol{p},\eta)$ is computable through an SDP. As
	\begin{align*}
	&\min_{\tau\in \tilde{\boldsymbol{D}}_A}\max_x\tilde{D}_{\alpha}(\tau||\rho_x)\\
	\le& \min_{\tau\in \tilde{\boldsymbol{D}}_A}\max_x \widehat{D}_{\alpha}(\tau||\rho_x)\\
	=&\min_{\tau\in \tilde{\boldsymbol{D}}_A}\inf_{\mu}\{\mu|\mu\ge \widehat{D}_{\alpha}(\tau||\rho_x),x=1,2,\cdots,k\}
	\end{align*}
	By applying the result shown in \cite{fang2021geometric}, we arrive at the following semidefinite representation, for $\alpha=1+2^{-l}$ with $l\in \mathbb{Z}^{+}$,
	\begin{align*}
	&-\ln(1-p_{succ}(\mathcal{E},\boldsymbol{p},\eta))+\frac{\alpha}{\alpha-1}(\ln(p_{min})+\ln(k-1+\eta))\\
	\le&\min_{\tau\in \tilde{\boldsymbol{D}}_A}\max_x\tilde{D}_{\alpha}(\tau||\rho_x)\\
	\le&\min_{\tau\in \tilde{\boldsymbol{D}}_A}\max_x\widehat{D}_{\alpha}(\tau||\rho_x)\\
	=&\inf \frac{1}{\alpha-1}\ln\mu\\
	\textit{s. t.} \hspace{3mm}&\mu\ge \tr M\\
	&\begin{pmatrix}
	M&X\\
	X&N_l
	\end{pmatrix}\ge 0, \\
	&\begin{pmatrix}
	X&N_i\\
	N_i&N_{i-1}
	\end{pmatrix}\ge 0, \hspace{3mm} i=1,2,\cdots,l,\\
	& N_0=\rho_x, x=1,2,\cdots,k,\\
	& N_1,N_2,\cdots,N_l\in \boldsymbol{Herm}_A.
	\end{align*}
	We can apply the result in \cite{fawzi2017lieb} to obtain $\widehat{D}_{\alpha}(\cdot||\cdot)$ in terms of a generic rantional number $\alpha\in(1,2].$
\end{remark}

Next we present an upper bound on the asymptotic success probability for the game $(\mathcal{E},\boldsymbol{p},\eta).$

\begin{theorem}
	Assume $\mathcal{E}^{(n)}=\{\rho_x^{\otimes n}\}_{x=1}^k$ is a set of $k$ quantum states with $n$ copies, Alice prepares a state $\rho_x^{\otimes n}$ in $\mathcal{E}^{(n))}$ with its prior probability $p_x$, the optimal success probabillity for  $(\mathcal{E}=\{\rho^{\otimes n}_x\},\boldsymbol{p}=\{p_x\},\eta)$ is
	
	\begin{align}
	limsup_{n\rightarrow+\infty}-\frac{1}{n}\ln(1- p_{succ}(\mathcal{E},\boldsymbol{p},\eta))\le \min_{\tau\in \tilde{\boldsymbol{D}}_A}\tilde{D}_{\alpha}(\tau||\rho_x),
	\end{align}
	here $\alpha\in [1,\infty).$
\end{theorem}
\begin{proof}
	In the theorem, we apply Theorem \ref{t2} to the ensemble $\mathcal{E}$, we have that
	\begin{align}
	&-\frac{1}{n}\ln (1-p_{succ}(\mathcal{E},\boldsymbol{p},\eta))\\
	\le& \inf_{\alpha\in[1,\infty)}\min_{\tau^{(n)}\in \tilde{\boldsymbol{D}}_{A^{\otimes n}}}\max_x\frac{1}{n}[\tilde{D}_{\alpha}(\tau^{(n)}||\rho_x^{\otimes n})-\frac{\alpha}{\alpha-1}(\ln(p_{min})+\ln(k-1+\eta))]\\
	\le& \inf_{\alpha\in[1,\infty)}\min_{\tau\in \tilde{\boldsymbol{D}}_{A}}\max_x\frac{1}{n}[\tilde{D}_{\alpha}(\tau^{\otimes n}||\rho_x^{\otimes n})-\frac{\alpha}{\alpha-1}(\ln(p_{min})+\ln(k-1+\eta))]\\
	\le& \inf_{\alpha\in[1,\infty)}\min_{\tau\in \tilde{\boldsymbol{D}}_A}\max_x\tilde{D}_{\alpha}(\tau||\rho_x)-\frac{\alpha}{n(\alpha-1)}(\ln(p_{min})+\ln(k-1+\eta)),
	\end{align}
	Next when $n\rightarrow+\infty$, the above inequality becomes
	\begin{align}
	&\lim\limits_{n\rightarrow +\infty}-\frac{1}{n}\ln(1-p_{succ}(\mathcal{E},\boldsymbol{p},\eta))\\\le& \inf_{\alpha\in[1,\infty)}\min_{\tau\in \tilde{\boldsymbol{D}}_A}\max_x\tilde{D}_{\alpha}(\tau||\rho_x)\\
	\le& \min_{\tau\in \tilde{\boldsymbol{D}}_A}\max_x{D}(\tau||\rho_x)\\
	=&R^{D}(\{\rho_x\}_x),
	\end{align}
	here the first inequality is due to that $\frac{\alpha}{\alpha-1}(\ln(p_{min})+\ln(k-1+\eta))$ is finite, the second inequality is due to the monotonicity of the $\tilde{D}_{\alpha}(\cdot||\cdot)$ in terms of $\alpha$ (Lemma \ref{l0}.3),
\end{proof}

\subsection{Quantum channel discrimination}
\indent In multiple quantum channel discrimination, here we assume there exists a device that can implement a channel from the set of channels $\mathcal{F}=\{\mathcal{N}_i\in \boldsymbol{C}_{A\rightarrow B}\}_{i=1}^k$, and each $\mathcal{N}_i$ happens with prior probability $p_i$, $i=1,2,\cdots,k$, here $p_i>0,$ $\sum_i p_i=1,$ let $\boldsymbol{p}=\{p_i\}_{i=1}^k.$ The player aims to discriminate the quantum channel with the optimal success probability. When $k=2$, the game reduces to the binary channel discrimination.

The game applies a general adaptive protocol for quantum channel discrimination. Specifically, the protocol starts with a bipartite state $\tau_{RA}$, then the player passes one system $A$ of $\tau_{RA}$ through some channel $\mathcal{N}_i$ from the set $\{\mathcal{N}_i\}_{i=1}^k$,  with ancillary global processing, the protocol ends with measuring the state after the final invocation. Here we plot Figure \ref{fig1} for an illustration. Formally, an adaptive protocol is denoted as 
\begin{align}
\mathcal{P}^{(n)}=(\tau_{RA},\mathcal{A}_{[n-1]},\mathcal{M}),
\end{align} 
here $\tau_{RA}$ is a bipartite state, $\mathcal{A}_{i}\in \boldsymbol{C}_{R_iB_i\rightarrow R_{i+1}A_{i+1}}$ is a channel, $n$ denotes the number of $\mathcal{N}_x$ produced in the protocol, and $\mathcal{M}=\{M_i\}_{i=1}^k$ is a POVM. We also assume $A_i\cong A$ and $B_i\cong B$, $i=1,2,\cdots,n$. Next we label the following states
\begin{align}
\rho_{x,1}=&\tau_{R_1A_1},\\
\sigma_{x,i}^{R_iB_i}=&\mathcal{N}_{x}(\rho_{x,i}^{R_iA_i})\hspace{3mm} i=1,2,\cdots,n,\\
\rho^{R_{i+1}A_{i+1}}_{x,i+1}=&\mathcal{A}_{i}(\sigma_{x,i}^{R_iB_i}) \hspace{3mm} i=1,2,\cdots,n-1.
\end{align}

Let $\sigma_{x,n}$ be the final state when the channel $\mathcal{N}_x$ is invoked with $n$ times, the optimal success probability of multiple channel discrimination is 
\begin{align}
p_{succ}^{(n)}(\mathcal{F}=\{\mathcal{N}_i\}_{i=1}^k,\boldsymbol{p}=\{p_i\}_{i=1}^k,\mathcal{P}^{(n)},\eta)=\sup_{\{M_i\}_{x=0}^{k}}\sum_{i=1}^kp_x\tr(M_x\sigma_{x,n}).
\end{align}
Here $\{M_i\}_{i=0}^{k}$ takes over all the POVMs with $M_0\ge\eta\mathbb{I}$, and we introduce the following quantity $p_{succ}^{(n)}(\mathcal{F},\boldsymbol{p},\eta)$,
\begin{align}
p_{succ}^{(n)}(\mathcal{F},\boldsymbol{p},\eta)=\sup_{\mathcal{P}^{n}}p_{succ}^{(n)}(\mathcal{F}=\{\mathcal{N}_i\}_{i=1}^k,\boldsymbol{p}=\{p_i\}_{i=1}^k,\mathcal{P}^{(n)},\eta),
\end{align}
here $\mathcal{P}^{(n)}$ takes over all the approciate adaptive protocol and

\begin{figure}
	\centering
	\includegraphics[scale=0.6]{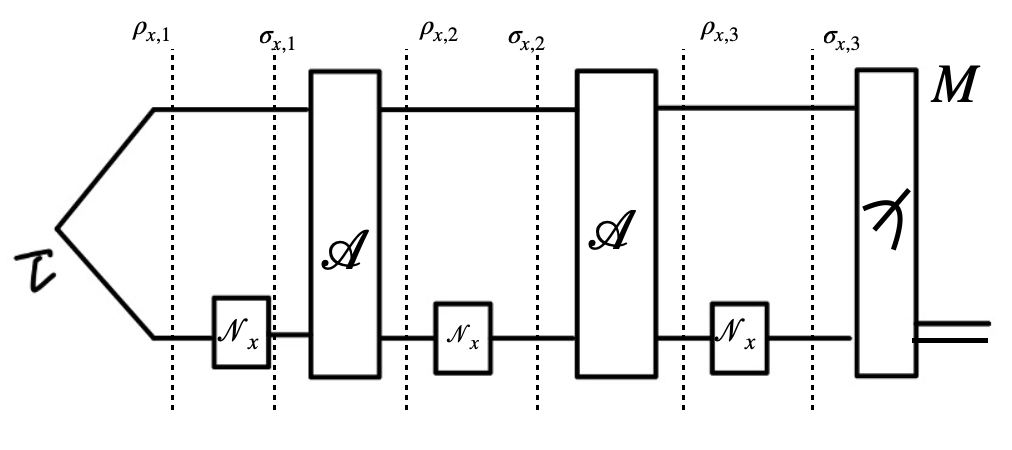}
	\caption{An adaptive strtegy for quantum channel discrimination when $n=3$.
		}\label{fig1}
\end{figure}

\begin{theorem}\label{t4}
	Let $\{\mathcal{N}_i\}_{i=1}^k$ be a set of quantum channels, a device implements a channel $\mathcal{N}_x$ with prior robability $p_i$, respectively. Then the optimal success probability of the revised unambiguous quantum channel discrimination is
	\begin{align}
	-\frac{1}{n}\ln(1-p_{succ}^{(n)}(\mathcal{F},\boldsymbol{p},\eta))\le \inf_{\mathcal{T}\in \boldsymbol{C}_{A\rightarrow B}}\sup_x\widehat{D}_{\alpha}(\mathcal{T}||\mathcal{N}_x)-\frac{\alpha}{n(\alpha-1)}(\ln(p_{min})+\ln(k-1+\eta)).
	\end{align}
	Here $\alpha\in(0,1].$
\end{theorem}
\begin{proof}
	Let $\mathcal{P}^{(n)}=(\tau_{RA},\mathcal{A}_{[n-1]},\mathcal{M})$ be the adaptive protocol, and let $\mathcal{T}\in \boldsymbol{C}_{A\rightarrow B}$ be an arbitrary channel from $\mathcal{H}_A$ to $\mathcal{H}_B$. Next we denote 
	\begin{align}
	\omega_1=&\tau\\
	\varepsilon_i=&\mathcal{T}(\omega_i) \hspace{3mm} i=1,2,\cdots,n\\
	\omega_{i+1}=&\mathcal{A}_i(\varepsilon_i)\hspace{3mm} i=1,2,\cdots,n-1.
	\end{align} 
	
	Then by applying Theorem \ref{t2}, we have
	\begin{align}
	&-\ln(1-p_{succ}^{(n)}(\mathcal{F},\boldsymbol{p},\mathcal{P}^{(n)},\eta))\nonumber\\
	\le&\inf_{\tau\in \tilde{\boldsymbol{D}}_A}\sup_x[\tilde{D}_{\alpha}(\tau||\sigma_{x,n})]-\frac{\alpha}{\alpha-1}(\ln(p_{min})+\ln(k-1+\eta))\nonumber\\
	\le& \inf_{\mathcal{T}\in \boldsymbol{C}_{A\rightarrow B}}\sup_x \tilde{D}_{\alpha}(\epsilon_n||\sigma_{x,n})-\frac{\alpha}{\alpha-1}(\ln(p_{min})+\ln(k-1+\eta))\nonumber\\
	\le& \inf_{\mathcal{T}\in \boldsymbol{C}_{A\rightarrow B}}\sup_x \widehat{D}_{\alpha}(\epsilon_n||\sigma_{x,n})-\frac{\alpha}{\alpha-1}(\ln(p_{min})+\ln(k-1+\eta)),\label{f0}
	\end{align}
	here the second inequality is due to the data processing inequlity of $\tilde{D}_{\alpha}(\cdot||\cdot)$ \cite{beigi}, the last inequality is due to $\widehat{D}_{\alpha(\cdot||\cdot)}\ge \tilde{D}_{\alpha}(\cdot||\cdot)$ \cite{fang2021geometric}.

	Next by applying the chain rule of the geometric Renyi channel divergence \cite{fang2021geometric,berta2022}, we have
	\begin{align}
	\widehat{D}_{\alpha}(\epsilon_i||\sigma_{x,i})\le \widehat{D}_{\alpha}(\mathcal{T}||\mathcal{N}_x)+\widehat{D}_{\alpha}(\omega_i||\rho_{x,i})\hspace{3mm} i=1,2,\cdots,n,\label{f1}\\
	\widehat{D}_{\alpha}(\omega_{i+1}||\rho_{x,i+1})\le \widehat{D}_{\alpha}(\omega_i^{'}||\rho_{x,i}^{'})\hspace{3mm} i=1,2,\cdots,n-1,\label{f2}
	\end{align}
	the first inequality is due to the definition of $\omega_n^{'}$ and $\rho_{x,n}^{'}$ and (\ref{f0}), the second inequality is due to the definition of $\omega$.
	\begin{align}
	\widehat{D}_{\alpha}(\epsilon_n||\sigma_{x,n})\le& n\widehat{D}_{\alpha}(\mathcal{T}||\mathcal{N}_x)+\widehat{D}(\omega_1||\rho_{x,1})\\
	=&n\widehat{D}_{\alpha}(\mathcal{T}||\mathcal{N}_x),\label{f3}
	\end{align}
	where the first inequality is derived by recursively applying (\ref{f1}) and (\ref{f2}). Combing (\ref{f3}) and (\ref{f0}), we have
	\begin{align}
	-\frac{1}{n}\ln(1-p_{succ}^{(n)}(\mathcal{F},\boldsymbol{p},\mathcal{P}^{(n)},\eta))\le \inf_{\mathcal{T}\in \boldsymbol{C}_{A\rightarrow B}}\sup_x\widehat{D}_{\alpha}(\mathcal{T}||\mathcal{N}_x)-\frac{\alpha}{n(\alpha-1)}(\ln(p_{min})+\ln(k-1+\eta)).\label{f4}
	\end{align}
	
	At last, as (\ref{f4}) holds for any adaptive strategy $\mathcal{P}^{(n)}$, 
	\begin{align}
	&-\frac{1}{n}\ln(1- p_{succ}^{(n)}(\mathcal{F},\boldsymbol{p},\eta))\nonumber\\
	=&\sup_{\mathcal{P}^{(n)}}\{-\frac{1}{n}\ln(1-p^{(n)}_{succ}(\mathcal{P}^{(n)},\eta)))|\mathcal{A}_i\in \boldsymbol{C}_{R_iA\rightarrow R_{i+1}B}\}\nonumber\\
	\le&\inf_{\mathcal{T}\in \boldsymbol{C}_{A\rightarrow B}}\sup_x\widehat{D}_{\alpha}(\mathcal{T}||\mathcal{N}_x)-\frac{\alpha}{n\alpha-n}(\ln(p_{min})+\ln(k-1+\eta)).
	\end{align}
	Hence, we finish the proof.
\end{proof}
\begin{remark}
	Here we present that the converse bound in Theorem \ref{t4} can be represented by an SDP, hence, the bound is efficiently computable. 
	\begin{align}
	&\inf_{\mathcal{T}\in \boldsymbol{C}_{A\rightarrow B}}\sup_x\widehat{D}_{\alpha}(\mathcal{T}||\mathcal{N}_x)\\
	=&\inf_{\mathcal{T}\in \boldsymbol{C}_{A\rightarrow B}}\inf_{\lambda}\{\lambda|\lambda\ge\widehat{D}_{\alpha}(\mathcal{T}||\mathcal{N}_x),x=1,2,\cdots,k\}.\label{f5}
	\end{align}
	Combing (\ref{f5}) with the result in \cite{fang2021geometric}, we have a SDP for the converse bound, let $\alpha=1+2^{-l}\in (1,2]$ with $l\in \mathbb{Z}^{+}$,
	\begin{align}
	\inf_{\mathcal{T}\in \boldsymbol{C}_{A\rightarrow B}}\sup_x\widehat{D}_{\alpha}(\mathcal{T}||\mathcal{N}_x)=\inf 2^{l}\ln\lambda&\\
	\textit{s. t. } \begin{pmatrix}
	M_x& J_{\mathcal{T}}\\
	J_{\mathcal{T}}& N_{l+1,x}
	\end{pmatrix}\ge&0, \hspace{3mm}
	\begin{pmatrix}
	J_{\mathcal{T}}&N_{i+1,x}\\
	N_{i+1,x}&N_{i,x}
	\end{pmatrix}\ge 0,\\
	N_{1,x}=J_{\mathcal{N}_x},\hspace{5mm}&
	\lambda\mathbb{I}_A\ge \tr_BM_{x,AB}, &\\
M_{x}\in \boldsymbol{Herm}_{AB},&\hspace{5mm}
N_{i+1,x}\in \boldsymbol{Herm}_{AB},\\	i=1,2,&\cdots,l,\hspace{5mm}
	x=1,2,\cdots,k.
	\end{align}
	For a general number $\alpha\in (1,2]$, a SDP can be obtained with the use of Theorem 3 in \cite{fawzi2017lieb,fang2021geometric}.
\end{remark}

\begin{Corollary}
	Let $\{\mathcal{N}_i\}_{i=1}^k$ be a set of quantum channels, a device implements a channel $\mathcal{N}_x$ with prior robability $p_i$, respectively,  then the optimal success probability of the revised unambiguous quantum channel discrimination is
	\begin{align}
	\lim\sup\limits_{n\rightarrow \infty}-\frac{1}{n}\ln(1-p_{succ}^{(n)}(\mathcal{P}^{(n)},\eta))\le \inf_{\mathcal{T}\in \boldsymbol{C}_{A\rightarrow B}}\sup_x\overline{D}(\mathcal{T}||\mathcal{N}_x).
	\end{align}
\end{Corollary}

\begin{proof}
	When $n\rightarrow\infty$, Theorem \ref{t4} turns into 
	\begin{align}
	&\lim\sup\limits_{n\rightarrow\infty}-\frac{1}{n}\ln(1-p_{succ}^{(n)}(\{\mathcal{E},\boldsymbol{p}\}_{i=1}^k,\mathcal{P}^{(n)},\eta))\nonumber\\
	\le&\inf_{\alpha\in(1,2]}\inf_{\mathcal{T}\in \boldsymbol{C}_{A\rightarrow B}}\sup_x \widehat{D}_{\alpha}(\mathcal{T}||\mathcal{N}_x)\\
	=&\inf_{\mathcal{T}\in \boldsymbol{C}_{A\rightarrow B}}\sup_x\overline{D}(\mathcal{T}||\mathcal{N}_x).
	\end{align}
	Here the first equalilty is due to the monotonicity of $\widehat{D}_{\alpha}(\cdot||\cdot)$ in terms of $\alpha$, and when $\alpha\rightarrow 1$, $\widehat{D}_{\alpha}(\cdot||\cdot)$ turns into the Belavkin-Staszewski divergence (\ref{bs}).
\end{proof}

\subsection{Advantage of quantum resource discrimiation}
\indent In this subsection, we first introduce a quantity on a set of positive semi-definite operators, then we present an operational interpretation on the function. By investigating the function, we obtain the advantage of the revised quantum state unambiguous discrimination compared with the classical.

Assume $\{Q_k\}_k$ is a set of semidefinite positive matrices,  $q_{re}(\{Q_k\})$ is defined as
\begin{align}
q_{re}(\{Q_k\})=&\min \mathrm{Tr}(P)\label{re 1}\\
\textit{s. t.} \hspace{3mm} & P \ge Q_k, \nonumber\\
& Q_k \ge 0\hspace{3mm} \forall k\nonumber.
\end{align}

We will show some properties of $q_{re}(\cdot)$ in the section \ref{app2}.

Next we consider a weaker quantum state unambiguous discrimination game. Assume Alice prepares $k$ states $\rho_x$ $ (x=1,2,\cdots,k)$ with a priori probability $p_x$, respectively, if the state produced is not known, then the state generated can be written as a mixed state
\begin{align}
\rho=\sum_xp_x\rho_x, \hspace{5mm} \textit{ with} \hspace{3mm}\sum_x p_x=1.
\end{align}
Bob can apply a POVM $\mathcal{M}=\{Q_1,Q_2,\cdots,Q_{k},Q_{\emptyset}\}$ with $(k+1)$ outcomes to detect Alice's state. In this scenario, Bob is allowed to waive the round when $Q_{\emptyset}$ happens. In this scenario, it is not necessary that the element $Q_k$ of $\mathcal{M}$ corresponds only to the state $\rho_k$.  When the detection outcome corresponding to $Q_{\emptyset}$ happens, the player will waive the round. Besides, the probability with the outcome corresponding to $M_{\emptyset}$ happens at least $\eta$, $i.$ $e.$ $Q_{\emptyset}\ge \eta\mathbb{I}$. The goal for Bob is to maximize the success probability of the game,
\begin{align}
P_{succ}^q(\mathcal{E},\boldsymbol{p},\eta)=&\max\sum_x p_x\mathrm{Tr}(\rho_xQ_x),
\label{ud1}\\
\textit{s. t.}\hspace{3mm}&  \sum_x Q_x\le (1-\eta)\mathbb{I},\nonumber\\
&Q_x\ge 0,\nonumber
\end{align}

At last, we present the maximal probability in the classical scenario, which is denoted as the unambiguous discrimination of classical information. In this scenario, the player can waive the round with probability at least $\eta$, hence the maximal success probability in classical scenario is 
\begin{align}
P_{succ}^c({\boldsymbol{p},\eta})=(1-\eta)\max_xp_x,
\end{align}

Now we  can present the operational interpretation of $q_{re}(\cdot)$.

\begin{theorem}\label{t3}
	Assume $\mathcal{E}=\{\rho_x\}_{x=0}^{n-1}$ is a set of quantum states, then there exists a quantity $\eta_{*}$ such that 
	\begin{align*}
	q_{re}(\{\rho_x\})	=\max_{p>0}\frac{P_{succ}^q(\mathcal{E},\boldsymbol{p},\eta)}{P_{succ}^c(\boldsymbol{p},\eta)},
	\end{align*}
	here the maximum is attained with $\{p_x=\frac{1}{k}\}$.
\end{theorem}

Here we place the proof of Theorem \ref{t3} in the section \ref{app2}. Due to the property of $q_{re}(\{\rho_x\}_{x=0}^{n-1}),$ $q_{re}(\{\rho_x\}_{x=0}^{n-1})\ge \max_x\tr\rho_x=1$, and there exists $\{\rho_x\}_{x=0}^{n-1}$ such that $q_{re}(\{\rho_x\}_{x=0}^{n-1})>1.$ That is, quantum games owns advantages over the classical in the unambiguous state discrimination games.

\section{Conclusion}
\indent In this manuscript, we studied a revised quantum state and channel discrimination under the nonasymptotic and asymptotic scenarios. Moreover, we presented the advantages of discrimination games between the quantum and classical.

In the study of quantum state discrimination, we mainly considered a revised quantum state unambiguous discrimination in the nonasymptotic and asymptotic scenarios. First, we presented an analytical expression of the success probability of the revised quantum state unambiguous discrimination, then we generalized the results to an upper bound of the success probability in terms of the sandwiched Renyi divergence $\tilde{D}_{\alpha}(\cdot||\cdot)$ in the asymptotic scenario. Furthermore, we loosened the bound which can be computed efficiently with the use of semidefinite programming. Then we generalized the method in the study of revised quantum state discrimination to the quantum channel discrimination under the adaptive strategy. Specifically, we presented an upper bound of the quantum channel discrimination under the multi-shot scenario and asymptotic scenario. Furthermore, we presented an efficient method to compute the success probability of the quantum revised channel discrimination under the asymptotic scenario. Finally, compared with classical unambiguous discrimination, we showed the advantages of the quantum by studying a function on a set of semidefinite positive operators.

\section*{Acknowledgements}
\label{sec:acknowledgements}
This work was supported by the National Natural Science Foundation of China (Grant No.12301580)

\bibliographystyle{IEEEtran}
\bibliography{ref}


\section{Appendix}
\subsection{Properties of Quantum Divergence}\label{app1}
In this subsection, we recall the properties of sandwiched Renyi-$\alpha$ divergence and geometric Renyi divergence. 

\begin{lemma}\label{l0}
	For all $\alpha\in (1,+\infty)$, the extended sandwiched Renyi divergence has the following properties,
	\begin{itemize}
		\item[1.][Data-processing inequality \cite{beigi}] Let $\mathcal{N}_{A\rightarrow B}$ be a positive and trace-non-increasing map. Let $\delta$ and $\sigma$ be Hermitian and positive semidefinite, respectively. 
		\begin{align}
		\tilde{D}_{\alpha}(\mathcal{N}(\delta)||\mathcal{N}(\sigma))\le \tilde{D}_{\alpha}(\delta||\sigma).
		\end{align}
		\item[2.][monotonicity in $\alpha$ \cite{frank}] Assume $1\le \alpha\le\beta\le \infty$, let $\delta$ and $\sigma$ be Hermitian and positive semidefinite, respectively.
		\begin{align*}
		\tilde{D}_{\alpha}(\delta||\sigma)\le D_{\beta}(\delta||\sigma),
		\end{align*}
		when $\alpha\rightarrow 1$, $\tilde{D}_{\alpha}(\rho||\sigma)\rightarrow D(\rho||\sigma)=\tr[\rho(\log\rho-\log\sigma)].$
		\item[3.] [Joint quasi-convexity \cite{mosonyi}] Let $\{\rho_x|x=1,2,\cdots,k\}$ be a set of positive semidefinite operators, and let $\{\sigma_x|x=1,2,\cdots,k\}$ be a set of Hermitian operators. Then
		\begin{align}
		\tilde{\boldsymbol{D}}_{\alpha}(\sum_x p_x\rho_x||\sum_xp_x\sigma_x)\le \max_x\tilde{\boldsymbol{D}}_{\alpha}(\rho_x||\sigma_x).
		\end{align}
	\end{itemize}
\end{lemma}

\begin{lemma}\label{le1}
	For all $\alpha\in (1,2]$, $\widehat{D}_{\alpha}(\cdot||\cdot)$ satisfies the following properties,
	\begin{itemize}
		\item[1.][Data-processing inequality \cite{fang2021geometric}] Let $\mathcal{N}_{A\rightarrow B}$ be a positive and trace-non-increasing map. Let $\delta$ and $\sigma$ be Hermitian and positive semidefinite, respectively. The following inequality holds for all $\alpha\in(1,2],$
		\begin{align}
		\widehat{D}_{\alpha}(\mathcal{N}(\delta)||\mathcal{N}(\sigma))\le \widehat{D}_{\alpha}(\delta||\sigma).
		\end{align}
		\item[2.] [\cite{keiji,fang2021geometric}] Assume $\rho$ and $\sigma$ be quantum states, let $\alpha\in(1,2]$, then
		\begin{align}
		D(\rho||\sigma)\le \tilde{D}_{\alpha}(\rho||\sigma)\le \widehat{D}_{\alpha}(\rho||\sigma),
		\end{align}
		here $D(\rho||\sigma)=\tr\rho\log\rho-\tr\rho\log\sigma$.
		
		\item[3.]	[monotonicity in $\alpha$ \cite{fang2021geometric}] Assume $1\le \alpha\le\beta\le \infty$, let $\delta$ and $\sigma$ be Hermitian and positive semidefinite, respectively.
		\begin{align*}
		\tilde{D}_{\alpha}(\delta||\sigma)\le D_{\beta}(\delta||\sigma).
		\end{align*}
		
		\item[4.] [addivity \cite{mishra2024optimal}] When $\alpha\in(1,2],$ $\rho_1$, $\rho_2$, $\sigma_1$ and $\sigma_2$ are quantum states,
		\begin{align}
		\widehat{D}_{\alpha}(\rho_1\otimes\rho_2||\sigma_1\otimes\sigma_2)=\widehat{D}_{\alpha}(\rho_1||\sigma_1)+\widehat{D}_{\alpha}(\rho_2||\sigma_2).
		\end{align}
		\item[5.] [Direct-sum property \cite{mishra2024optimal}] Assume $\rho_{XA}$ and $\sigma_{XA}$ are two $c-q$ states,
		\begin{align*}
		\rho_{XA}=\sum_xp_x\ket{x}\bra{x}\otimes\rho_A^x,\\
		\sigma_{XA}=\sum_xp_x\ket{x}\bra{x}\otimes\sigma_A^x,
		\end{align*}
		let $\widehat{Q}_{\alpha}(\rho||\sigma)=\exp[(\alpha-1)\widehat{D}_{\alpha}(\rho||\sigma)],$	then
		\begin{align}
		\widehat{Q}_{\alpha}(\rho_{XA}||\sigma_{XA})=\sum_xp(x)^{\alpha}q(x)^{1-\alpha}\widehat{Q}_{\alpha}(\rho_A^x||\sigma_A^x).
		\end{align}
		\item[6.] [Joint convexity \cite{wilde}] Let $\rho$ and $\sigma$ can be written as $\rho=\sum_x p_x\varphi_x$ and $\sigma=\sum_x p_x\varsigma_x$, respectively, here $p_x\ge 0$ and $\sum_xp_x=1,$
		\begin{align}
		\widehat{D}_{\alpha}(\sum_xp_x\rho_A^x||\sum_xp(x)\sigma_A^x)\le\sum_x p_x\widehat{D}_{\alpha}(\rho_A^x||\sigma_A^x),\label{r2}
		\end{align}
		moreover, (\ref{r2}) can be relaxed to 
		\begin{align}
		\widehat{D}_{\alpha}(\sum_xp_x\rho_A^x||\sum_xp(x)\sigma_A^x)\le \max_x\widehat{D}_{\alpha}(\rho_A^x||\sigma_A^x).
		\end{align}
	\end{itemize}
\end{lemma}

\begin{lemma}\cite{fang2021geometric,mishra2024optimal}
	Assume $\mathcal{N}$ is a channel and $\mathcal{M}$ is completely positive. When $\alpha\in(1,2]$,
	\begin{itemize}
		\item when $1< \alpha\le \beta\le 2$, 
		\begin{align}
		\widehat{D}_{\alpha}(\mathcal{N}||\mathcal{M})\le \widehat{D}_{\beta}(\mathcal{N}||\mathcal{M}).
		\end{align}
		\item let $\{\mathcal{N}_x\}$ is a set of quantum channels, and let $\{\mathcal{M}_x\}$ is a set of completely positive maps,
		\begin{align}
		\widehat{D}_{\alpha}(\sum_xp_x\mathcal{N}_x||\sum_xp_x\mathcal{M}_x)\le \max_x\widehat{D}_{\alpha}(\mathcal{N}_x||\mathcal{M}_x),
		\end{align}
		here $p_x>0$ and $\sum_xp_x=1.$
		\item when $\alpha\rightarrow 1$, the geometric Renyi channel divergence turns into the Belavkin-Staszewski channel divergence,
		\begin{align}
		\widehat{D}(\mathcal{N}||\mathcal{M})=&	\lim\limits_{\alpha\rightarrow 1}\widehat{D}_{\alpha}(\mathcal{N}||\mathcal{M})\\
		=&\begin{cases}
		||\tr_B[\mathcal{J}_N^{\frac{1}{2}}\ln(\mathcal{J}^{\frac{1}{2}}_{\mathcal{N}}\mathcal{J}^{-1}_{\mathcal{M}}\mathcal{J}^{\frac{1}{2}}_{\mathcal{N}})\mathcal{J}^{\frac{1}{2}}_{\mathcal{N}}]||_{\infty}\hspace{5mm} \textit{if $supp(\mathcal{J}_{\mathcal{N}}) \subseteq supp(\mathcal{J}_{\mathcal{M}})$,}\\
		+\infty \hspace{5mm} \textit{otherwise}.
		\end{cases}
		\end{align}
		\item  when $\alpha\in(1,2],$
		\begin{align}
		\widehat{D}_{\alpha}(\mathcal{N}(\rho_{RA})||\mathcal{M}(\sigma)_{RA})\le\widehat{D}_{\alpha}(\mathcal{N}||\mathcal{M})+\widehat{D}_{\alpha}(\rho_{RA}||\sigma_{RA}),
		\end{align}
		here $\mathcal{N}\in \boldsymbol{C}_{A\rightarrow B}$, $\mathcal{M}\in \boldsymbol{CP}_{A\rightarrow B}$, $\rho\in \boldsymbol{D}_{RA}$, and $\sigma_{RA}\in\boldsymbol{CP}_{RA}.$
	\end{itemize} 
\end{lemma}

\begin{lemma}[\cite{mishra2024optimal}]\label{l4}
	Let $\rho$ be a Hermitian operator of the system $\mathcal{H}_A$ with $\tr\rho=1$, and let $\sigma$ be a state of system $\mathcal{H}_A$. Then 
	\begin{align}
	\tilde{D}_H^{\epsilon}(\rho||\sigma)\le	\tilde{D}_{\alpha}(\rho||\sigma)-\frac{\alpha}{\alpha-1}\ln(1-\epsilon).
	\end{align}
	Here $\epsilon\in(0,1).$
\end{lemma}
\begin{proof}
	Here the proof is a generalization of the relation between $D_H^{\epsilon}(\rho||\sigma)$ and $\tilde{D}_{\alpha}(\rho||\sigma)$ shown in \cite{}. Next we present the proof.
	
	Assume $Q$ is positive semidefinite with $Q\le \mathbb{I}$ and $\tr Q\rho\le\epsilon$. Let 
	\begin{align*}
	\mathcal{N}(\sigma)=\tr Q\rho\ket{1}\bra{1}+\tr (\mathbb{I}-Q)\sigma \ket{2}\bra{2}.
	\end{align*}
	Then by the Data-processing inequality \cite{beigi}, we have
	\begin{align*}
	\tilde{D}_{\alpha}(\rho||\sigma)\ge& \tilde{D}_{\alpha}(\mathcal{N}(\rho)||\mathcal{N}(\sigma))\\
	=&\frac{1}{\alpha-1}\ln[|\tr(Q\rho)\tr(Q\sigma)^{\frac{1-\alpha}{\alpha}}|^{\alpha}+|\tr((\mathbb{I}-Q)\rho)\tr((\mathbb{I}-Q)\sigma)^{\frac{1-\alpha}{\alpha}}|^{\alpha}]\\
	\ge& \frac{\alpha}{\alpha-1}\ln(\tr((\mathbb{I}-Q)\rho))-\ln(\tr((\mathbb{I}-Q)\sigma))\\
	\ge& \frac{\alpha}{\alpha-1}\ln(1-\epsilon)-\ln\tr(\mathbb{I}-Q)\sigma,
	\end{align*}
	here the second inequality is due to that $\tr(\mathbb{I}-Q)\rho\ge 1-\epsilon.$ For any $\Lambda\in \boldsymbol{PSD}_A$ satisfying $\Lambda\le \mathbb{I}$ and $\tr(\Lambda\rho)\ge 1-\epsilon$,
	\begin{align}
	\tilde{D}_H^{\epsilon}(\rho||\sigma)=&\max_Q\{-\ln(1-\tr Q\sigma)|\tr Q\rho\le \epsilon,0\le Q\le \mathbb{I} \}\\
	\le&\tilde{D}_{\alpha}(\rho||\sigma)-\frac{\alpha}{\alpha-1}\ln(1-\epsilon),
	\end{align}
	hence, we finish the proof.
\end{proof}

\subsection{Properties of $q_{re}(\{Q_k\})$}\label{app2}
\begin{lemma}
	Assume $\mathcal{E}=\{Q_x\}$ is a set of positive semidefinite operators acting on $\mathcal{H}_d$, then 
	\begin{itemize}
		\item[1]. $q_{re}(\{Q_x\})\ge  \max_x\tr Q_x$ , and the equality is attained if $Q_x=Q_y$ for each $x\ne y.$
		\item[2]. Assume $\mathcal{C}$ is an arbitrary channel, and when all $Q_x=\rho_x\in \boldsymbol{D}(\mathcal{H}),$ then $q_{re}(\{\mathcal{C}(\rho_x)\})\le q_{re}(\{\rho_x\}).$
	\end{itemize}
\end{lemma}
\begin{proof}
	\begin{itemize}
		\item[1]. Due to the defnitions of $q_{re}(\{Q_x\})$, we have $P\ge  \rho_k,$ $i.$ $e.$
		\begin{align*}
P\ge Q_x\Rightarrow& \tr P\ge \tr Q_x,\forall x\\
\Rightarrow& \tr P\ge \max_x \tr Q_x.
		\end{align*}
		
	Next we show when $Q_x=Q_y=M,$ $\forall x\ne y,$ $q_{re}(\{Q_x\})=\tr M.$	When $Q_x=Q_y=M,$ $\forall x\ne y,$ let $P$ in (\ref{re 1}) be $M$, then $q_{re}(\{Q_x\})\le  \tr M.$ Hence $q_{re}(\{\rho_x\})=\max_x \tr M.$
		\item[2].  Assume $\mathcal{C}$ is an arbitrary channel, 
		\begin{align*}
		\mathcal{C}(\rho)=\sum_k E_k^{\dagger}\rho E_k,
		\end{align*}
		here $\sum_k E_k E_k^{\dagger}=\mathbb{I}.$
		\begin{align*}
		q_{re}(\{\rho_x\})=&\min \tr P\\
		\textit{s. t.}\hspace{3mm}& P\ge \rho_x,
		\end{align*}
		let $P$ be the optimal for $\{\rho_x\}$ in terms of $q_{re}(\cdot)$, $$\mathcal{C}(P)=\sum_k E_k^{\dagger}P E_k\ge \sum_k E_k^{\dagger}\rho_x E_k,\forall x,$$
		that is, $\mathcal{C}(P)$ is appropriate for $\{\rho_x\}$ in terms of $q_{re}(\{\rho_x\})$, then 
		\begin{align*}
		q_{re}(\{\rho_x\})=&\tr (P)\\
		=&\tr\sum_k E_kE_k^{\dagger}P \\
		\ge& q_{re}(\{\mathcal{C}(\rho_x)\}),
		\end{align*}
		here the first equality is due to $\sum_k E_kE_k^{\dagger}=\mathbb{I},$ the first inequality is due to that $\mathcal{C}(P)$ is appropriate for $\{\mathcal{C}(\rho_x)\}$ in terms of $q_{re}(\cdot)$.
	\end{itemize}
\end{proof}

Next we show the duality of $q_{re}(\{Q_k\})$,

\begin{align}
\max &\sum_i Y_iQ_i\label{re 2}\\
\textit{s.t.}&\hspace{3mm} \sum_i Y_i\le \mathbb{I},\nonumber\\
&\hspace{3mm}Y_i\ge 0\nonumber
\end{align}
\begin{proof}
	Consider the Lagrangian
	\begin{align*}
	&\mathcal{L}(P,\{Q_x\}_x,R)\\
	=&\tr(P)+\sum_x\tr((Q_x-P)M_x)\\
	=&\sum_x\tr(Q_xY_x)+\tr(P(I-\sum_xY_x)),
	\end{align*}
	then the dual SDP is 
	\begin{align*}
	q_{re}(\{Q_k\})=&\max_{} \sum_k \tr(Q_kY_k)\\
	\textit{s. t.}&\hspace{3mm} \sum_k Y_k\le \mathbb{I},\\
	&\hspace{3mm} Y_k \ge 0.
	\end{align*}
	At last, we show strong duality holds for the SDP (\ref{re 1}). Note that $P=\sum_i\rho_i$ is a feasible solution to the primal program. For the dual program, when $Y_1=\frac{1}{2}\ket{0}\bra{0},$ $Y_i=0,$ $\forall i\ne 1,$ then $\mathbb{I}-\sum_i Y_i=\mathbb{I}-\frac{1}{2}\ket{0}\bra{0}$. Let $\ket{v}$ be an arbitrary state, $\bra{v}(\mathbb{I}-\sum_i Y_i)\ket{v}>0$.  Hence  the Slater's condition are satisfied, we finish the proof.
\end{proof}

Next we present the following expression of $q_{re}(\cdot)$.
\begin{lemma}\label{l2}
	There exists $\eta_{*}<1$ such that $\forall\eta\in[\eta_{*},1]$,
	\begin{align*}
	(1-\eta)q_{re}(\{Q_x\})=&\max  \sum_x\tr(Y_xQ_x)\\
	\textit{s. t.}\hspace{3mm}& \sum_x Y_x\le (1-\eta)\mathbb{I},\\
	& Y_x\ge 0,\forall x
	\end{align*}
\end{lemma}
\begin{proof}
	Assume $\{Y_x\}$ is the optimal for $q_{re}(\{Q_x\})$ in terms of (\ref{re 2}) with $\sum_x Y_x\le \mathbb{I}.$ Let $\eta_{*}=1-||\sum_xY_x||_{\infty}$.
	
	\begin{align*}
	(1-\eta)q_{re}(\{Q_x\})=\max \sum_k\tr[(1-\eta)Q_kY_k],
	\end{align*}
	Let $Y_x^{*}=(1-\eta)Y_x$, then 
	\begin{align*}
	(1-\eta)q_{re}(\{Q_x\})=&\max\sum_x\tr[Y_x^{*}Q_x]\\
	\textit{s. t.}\hspace{3mm}& \sum_xY^{*}_x\le (1-\eta)\mathbb{I},\\
	&	Y_x^{*}\ge 0.
	\end{align*}
\end{proof}

Then we show the duality of $P_{succ}^{q}(\mathcal{E},\boldsymbol{p},\eta),$

Consider the following Lagrangian 
\begin{align*}
\mathcal{L}(M,\{p_x\},\{\rho_x\},\eta)=& \sum_x p_x\tr(\rho_x Q_x)+\tr[(1-\eta)\mathbb{I}-\sum_xQ_x]M\\
=&  (1-\eta)\tr M+\sum_x \tr(p_xQ_x-M)
\end{align*}

hence, the dual of the above SDP (\ref{ud1}) can be written as
\begin{align*}
P_{succ}^q(\mathcal{E},\boldsymbol{p},\eta)=&\min (1-\eta)\tr M\\
\textit{s. t.}\hspace{3mm}& M\ge p_x\rho_x,
\end{align*}

At last, we show the following theorem.
\begin{theorem}
	Assume $\mathcal{E}=\{\rho_x\}_{x=0}^{n-1}$ is a set of quantum states, then there exists a quantity $\eta_{*}$ such that 
	\begin{align*}
	q_{re}(\{\rho_x\})	=\max_{p>0}\frac{P_{succ}^q(\mathcal{E},\boldsymbol{p},\eta)}{P_{succ}^c(\boldsymbol{p},\eta)},
	\end{align*}
	here the maximum is attained with $\{p_x=\frac{1}{k}\}$.
\end{theorem}
\begin{proof}
	\begin{align*}
	q_{re}(\{\rho_x\})=&\max_{p>0}\frac{\max\limits_{\sum_x  Y_x\le(1-\eta)\mathbb{I},
			Y_x\ge 0\forall x}\sum_x\max_yp_y\tr(Y_xQ_x)}{\max_y p_y(1-\eta)}\\
	\ge&\max_{p>0}\frac{\max\limits_{\sum_xY_x\le(1-\eta)\mathbb{I},
			Y_x\ge 0\forall x}\sum_xp_x\tr(Y_xQ_x)}{\max_y p_y(1-\eta)}
	\end{align*}
	The first equality is due to Lemma \ref{l2}. As the inequality can be saturated by the following probability $\{p_x=\frac{1}{d}\},$ we finish the proof. 
\end{proof}

\end{document}